\documentclass[12pt, a4paper]{amsart}

\usepackage{hyperref}
\usepackage{amsmath, amsfonts, amsthm,bm,mathrsfs}
\usepackage{amsmath}
\usepackage{natbib}
\usepackage{graphicx}
\usepackage{amscd,amssymb,amsthm}
\makeatletter
\g@addto@macro{\endabstract}{\@setabstract}
\newcommand{\authorfootnotes}{\renewcommand\thefootnote{\@fnsymbol\c@footnote}}
\makeatother

\theoremstyle{plain}
\newtheorem{theorem}{Theorem}[section]

\newtheorem{proposition}[theorem]{Proposition}

\theoremstyle{remark}

\begin{document}
\bibliographystyle{plainnat}

\begin{center}
\LARGE Estimating the turning point  location in  shifted exponential model of  time series
\vskip20pt

 \normalsize
  \authorfootnotes
  Camillo Cammarota\footnote{cammar@mat.uniroma1.it} \\
  Department of Mathematics  University La Sapienza,\\
 P.le A. Moro 5, 00185 Rome, Italy
 
 \par \bigskip

  \today
\end{center}

\begin{abstract}
We consider the distribution of the turning point  location of    time series modeled as the  sum of   deterministic trend  plus random  noise. If the variables are modeled by    shifted exponentials,  whose location parameters define the trend, we provide  a formula for computing the  distribution of the turning point location and consequently   to estimate a confidence interval for the location. 
We  test this formula in   simulated data series  having a trend with    asymmetric minimum, investigating the coverage rate as a function of a bandwidth parameter. 
The method  is applied to estimate the  confidence interval of the minimum location of the  time series of  RT   intervals  extracted from   the  electrocardiogram  recorded during the exercise test.  We discuss the connection  with stochastic ordering.   \end{abstract}


\section{Introduction}

Real data time series  are usually modeled  as the sum of a deterministic sequence (trend) plus  a stochastic component (noise).  These series   often exhibit non monotonic trends   with one or more   turning points (maximum or minimum). One of the problems mostly investigated  is   on line signaling  of the turning point, for a timely detection of the regime switch.  Various methods have been adopted to test the occurrence of these events; we refer to \cite{andersson} for a review.

A different perspective concerns the  precision of the location of the turning point for a fixed series of observations.  The location  time can be considered a random variable whose distribution has to  be estimated from the data.
This problem was considered in a parametric  approach  assuming that  the trend   close to a turning point  is a second order polynomial. In this case   there are explicit formulas for the location   distribution as a function of the  distributions of the estimated parameters. Consequently  confidence intervals  can be estimated as in  \cite{mudambi, plassmann}; more recently confidence regions for the  minimum  location  in multivariate case are estimated in \cite{sambucini}. 
 
    The detection and estimate of turning points has  different characteristics  with respect to   change point  analysis.  A distinctive  feature of turning points is that   the trend changes continuously.  Conversely in a change point  the trend  of the series  changes abruptly;  for a review of the intensive literature on the subject we refer to \cite{jandhyala}.  Recently  in the case of multiple change points  in which the trend  of the series is a stepwise function, confidence band for the jumps and  confidence intervals for their locations were estimated  in \cite{frick}. 

In  time series from several areas the assumption  of a polynomial trend is not correct, for instance   when the  
 extremum is non symmetric, with left and right trends  having different slopes or  convexities. This is the case of the series here considered, obtained from the measurement of time intervals in cardiac monitoring. In these and  many other cases  non parametric methods for estimating the trend  are required,  for which we refer to \cite{ruppert, wasserman}.      These methods   allow to estimate a confidence band for the trend  and consequently  also a confidence interval for the  value of the turning point,  but they   do not provide a confidence interval for its location. 
 The estimate of this  interval is an open problem and  it is the main motivation  of the present paper.
   
We consider  the case  of series  modeled by shifted exponentials: the trend is modeled by a sequence of location parameters and   the stochastic component  by  a sequence of  standard exponentials with constant rate. 
 
 This model is motivated by the  analysis of series of time intervals in cardiac monitoring. 
 The   occurrence times of the   R peaks  in the electrocardiogram (ECG) are  modeled as  a  point process   and various distributions have been proposed to model the interbeat  RR intervals, which measure the duration of the cardiac cycle \cite{barbieri}. The existence of a functional refractory period suggests that   the distribution of  the interbeat  RR interval   can be modeled  by a  shifted exponential, as proposed in \cite{corino}. During the  stress test the series of RR intervals shows a clear trend having a global minimum. The decreasing region  close to the minimum has usually a different slope with respect to the increasing one \cite{cammarota2011-1}.
Strictly related to the previous  one is  the series of RT intervals, defined as the interval from the R peak to the apex of T wave, corresponding to the duration of the repolarization phase of the cardiac cycle.   Both RR and RT series during the stress test have a global minimum. Despite of the importance  of   these series  in clinics (see for instance \cite{lauer}),   the problem of the minimum location was  only recently  addressed in  \cite{cammarota2012}.  

The paper is organized as follows.
In sec. 2 we provide an integral    representation   of  the  distribution of the minimum location  as a function of the trend, from which confidence intervals can be obtained.   In sec. 3 we  discuss  a connection with   stochastically ordered variables. 
In sec 4 we check the confidence  interval on simulated series having  non symmetric minimum,  computing numerically the coverage rate  of the  interval. In sec. 4 we apply the method to a real data series obtained from cardiac monitoring. In sec. 6 we provide some  conclusions.

\section{The  turning point location}
 
For the time series $Y_t, \ t\in I$  were $I$ is  a finite interval of integers, we assume  the  model

\begin{equation}
Y_t= T_t +\epsilon_t
\label{trend}
\end{equation}
where $T_t$ denotes  a deterministic sequence (trend) and $\epsilon_t$ a  sequence of  continuous random variables (noise).   The sequence    $\epsilon_t$ is  assumed to be   stationary  and independent   with regular density  denoted $f(u), u\in \mathbb{R}$. We consider  the case in which the trend has a minimum, since the maximum case can be treated analogously.
The location of the minimum is defined as the random variable $\tau$  on $I$ 

\begin{equation}
 \tau = {\rm argmin} \{ Y_t, \  t \in I \}
 \label{tau}
 \end{equation}
 Consequently the event that the minimum is attained at the index $s$ is 
\begin{equation}
\{\tau=s\}= \bigcap_{t:\,t\ne s} \{ Y_t > Y_s \} , \quad s\in I
 \end{equation}
Using the independence of the $\epsilon _t$,  the discrete density on $I$ of the r.v. $\tau$ is 

\begin{equation}
P(\tau=s\ ) = \prod_{t:\,t\ne s} P( Y_t > Y_s), \quad s\in I
\label{prod}
 \end{equation}
where  the probability that the minimum is attained at more than one point  is zero.

In eq. \ref{trend} we assume  the $Y_t$ to be  shifted exponential variables   with location parameters $T_t$ and 
$\epsilon_t$=Exp$(\lambda)$ so that 

\begin{equation}
P(Y_t>y) = e^{- \lambda \ (y-T_t) \ \chi_{[T_t, +\infty)} (y)}, \quad t\in I
\label{shift}
\end{equation}
Here   the rate parameter $\lambda$ is constant and $\chi_A$ denotes the characteristic function of $A$.  
Our main result is the following one.

\begin{proposition}
For any  trend $T_t, t\in I, $ the distribution of the minimum location  of  the sequence defined by  eq. \ref{shift}
is 
\begin{equation}
P(\tau=s) = \int_{T_s}^{+\infty} du\  \lambda\  e^ {- \lambda \displaystyle\sum_{t:\, T_t \le u}( u-T_t)}\ , \quad s\in I
\label{expB}
\end{equation}
\end{proposition}

\begin{proof}
 Denoting $f_t$ the density of $Y_t$, i.e. 
\begin{equation}
f_t(u) =  \lambda\  e^{- \lambda ( u- T_t)} \chi_{[T_t, +\infty)}(u)
\label{density}
\end{equation} 
we have
$$
P(\tau= s)=  \int_{-\infty}^{+\infty}  du \ f_s (u) \prod_{t:\,t\ne s} e^{- \lambda\ (u -T_t) \ \chi_{[T_t, +\infty)} (u)} $$
$$
= \int_{-\infty}^{+\infty}  du \ f_s (u)\  e^{- \lambda \displaystyle\sum_{t:\,t\ne s}  (u -T_t) \ \chi_{[T_t, +\infty)} (u)} 
$$
From
$$
\sum_{t:\, t\ne s}  (u -T_t) \ \chi_{[T_t, +\infty)} (u) = \sum_{t:\, t\ne s, T_t \le u}  (u -T_t) 
$$
and  from eq. \ref{density} the last integral is
$$
 \int_{-\infty}^{+\infty}  du \   \lambda\  e^{- \lambda ( u- T_s)} \chi_{[T_s, +\infty)}(u)\    e^{-\lambda \displaystyle\sum_{t:\, t\ne s,\, T_t \le u}  (u -T_t)}
 $$
and the result follows.
\end{proof}
Denoting 
\begin{equation}
B(u) =  \sum_{t: \, T_t \le u}  (u -T_t)
\label{B}
\end{equation}
the eq. \ref{expB} reads
\begin{equation}
P(\tau=s) = \int_{T_s}^{+\infty} du\  \lambda\  e^ {- \lambda B(u)}
\label{expB1}
\end{equation}

The function $B(u)$ plays a central role as it measures the area delimitated by the trend $T_t$ below the horizontal line of height $u$. Let us  assume  for instance that $T_t$ is a regular  convex trend having a minimum.  
If the index $s$ is far from the  minimum  location so that $T_s$ is well above  the minimum value, 
since the integration domain is $[T_s, +\infty)$,  the  argument $u$ in $B(u)$ is   large   and consequently  $B(u)$ is large and  $P(\tau = s)$ is small. 

\section{Asymmetric  minimum}

We illustrate the  use of Proposition 2.1 in the  analysis of the model  in eq. \ref{trend}  when the trend  $T_t$ is estimated parametrically and its functional dependence on the $t$ index is known.
We always assume  that the   noise   $\epsilon_t= {\rm Exp}(\lambda)$  is an  i.i.d.   sequence of exponentials with mean $1/\lambda$.  Eq. \ref{expB1} provides a representation of the minimum location distribution in closed form, from which exact computations can be performed. We first consider a piecewise linear symmetric  trend having  the minimum at $t=0$: 
\begin{equation}
T_t=a |t|, \quad a>0
\label{trend1}
\end{equation}
We have 
$$B(u) =  \sum_{t: T_t \le u}( u-T_t) = \sum_{t:a|t| \le u} ( u - a|t|) \simeq \int_{-u/a}^{u/a } (u - a|t|)\  dt = \frac{u^2}{a}$$
and from  eq. \ref{expB1} one gets
$$ 
P(\tau=s) = \int_{a|s|}^{+\infty} \lambda\  e^{- \lambda\  \frac{u^2}a}\  du = \int _1^{+\infty} \lambda \ e^{- \lambda a s^2 v^2} dv
$$
The last integral shows  the dependence on  the slope parameter $a$, i.e.  the more deep  is the minimum, the  more concentrated  is the distribution of $\tau$, as expected.  This property is similar to the one showed in \cite{frick} in which   larger  jumps  of the  stepwise trend are associated to smaller  location intervals. 

We consider  an example of  time series  with  non symmetric minimum, i.e. having different   right and left slopes. 
\begin{equation}
T_t=\begin{cases} \quad  bt,  & \quad  t\ge 0 \\
-at, & \quad  t<0. \end{cases}
\label{trend2}
\end{equation}
where we assume $0<a<b$. Using the same argument as above one has
\begin{equation}
B(u) = \frac{u^2}{2a} + \frac{u^2}{2b}
\label{asymm}
\end{equation}
Using eq.  \ref{expB1} the distribution  of $\tau$ is given by 
 \begin{equation}
 P(\tau=s) = \begin{cases} \int_{bs}^{+\infty} du\  \lambda\  e^ {- \lambda B(u)}& \quad  s\ge 0 \\
\int_{-as}^{+\infty} du\  \lambda\  e^ {- \lambda B(u)}& \quad  s < 0. \end{cases}
\end{equation}
The expectation of $\tau$ is 
$$\mathbb{E} (\tau) =
\sum_{s} s P(\tau=s) = \sum_{s>0} s \int_{bs}^{+\infty} du\  \lambda\  e^ {- \lambda B(u)} + 
\sum_{s<0} s \int_{-as}^{+\infty} du\  \lambda\  e^ {- \lambda B(u)} 
$$
The second summand can be written as 
$$
-\sum_{s>0} s \int_{as}^{bs} du\  \lambda\  e^ {- \lambda B(u)} - \sum_{s>0} s \int_{bs}^{+\infty} du\  \lambda\  e^ {- \lambda B(u)}
$$
and so  one has
$$
\sum_{s} s P(\tau=s) = - \sum_{s>0} s \int_{as}^{bs} du\  \lambda\  e^ {- \lambda B(u)}
$$
Using the  integration variable $v=u/s$  and eq. \ref{asymm},  the previous term is 
$$
-  \sum_{s>0} s^2  \int_{a}^{b} dv\  \lambda e^{ -\lambda s^2 (\frac{v^2}{2a} + \frac{v^2}{2b})}
$$
We  firstly perform the sum over $s$   approximating   it with the  Gaussian integral 
$$\sum_{s>0} s^2 e^{- k s^2} \approx \frac12\frac{\sqrt{2\pi}}{(2k)^{3/2}}, \quad k>0$$
and then we compute the integral  over $v$ 
$$ \int_a^b \ dv \ \lambda \frac12 \frac {\sqrt{2\pi}}{ (2\lambda(\frac{v^2}{2a} + \frac{v^2}{2b}))^{3/2}}$$
which finally  gives the  formula
\begin{equation}
 \mathbb{E} (\tau) = - \frac{\sqrt{2\pi}}{\sqrt \lambda}\frac{b-a}{\sqrt{ab(a+b)}}
 \label{bias}
\end{equation}
Notice that the only  approximation  we have done consists in replacing sums with integrals.
This results shows that the natural estimator $\mathbb{E} (\tau)$ for the minimum location is biased.  In particular it is biased towards the region with smaller slope, in the example   towards negative indices since $a<b$.

\section{Stochastic ordering}

The existence of a  deterministic trend is strictly related to the stochastic ordering between the variables.
The notion of stochastic order is used  in many applications (for a general reference  in reliability theory  see \cite{shaked}) but, at the best of our knowledge, it is rarely used in time series.  
We recall that   the r.v. $Y_1$ is stochastically smaller than $Y_2$,  written $Y_1 \prec Y_2$,  if   for any $y\in  \mathbb{R}$ one has 
$$P(Y_1>y) \le  P(Y_2 >y)$$
 For a monotonic increasing trend,  $T_t < T_u$ if $ t < u$,  the sequence  $Y_t$ is stochastically increasing. Actually one has 
$$P(Y_t >y) = P( T_t + \epsilon _t > y) <   P( T_u + \epsilon _t > y)= P( T_u+ \epsilon _u > y)= P(Y_u >y)$$
The  shifted exponentials  provide   an example of stochastically ordered  variables if  the location parameters  form  a monotonic sequence.
A remarkable example in which the distribution of $\tau$ can be computed  exactly is when the 
variables  $Y_t$ are such that 
\begin{equation}
P(Y_t > y) = P(Y>y)^{\alpha_t} 
\label{galfa}
\end{equation}
for some suitable  variable $Y$   and some positive sequence $\alpha_t$.
If  $\alpha_t $ is decreasing then the  sequence $Y_t$ is stochastically increasing. 
Actually, denoting  $G(y) = P(Y>y), \ G_t(y) =P(Y_t>y)$, with $G(y)>0$, 
if  $\alpha_t > \alpha_u$ for  $ t < u$,  one has 
$$ G_t(y) =  G(y) ^{\alpha_t} <   G(y) ^{\alpha_u}= G_u(y)$$
\begin{proposition}
If a time series  $Y_t, \ t\in I,$  satisfies eq.  \ref{galfa}, then the distribution of the minimum location is:
\begin{equation}
P(\tau=s)= \frac{\alpha_s}{\sum_{t\in I}\ \alpha_t}\ , \quad s\in I
\label{alphaoversum}
\end{equation}
\end{proposition}
\begin{proof}
We denote $f_t$ the density of $Y_t$ and  using eq. \ref{prod} we get
$$
P(\tau=s) = \int_{-\infty}^{+\infty}  du \ f_s (u) \prod_{t:\, t\ne s} \int_u^{+\infty}  dv f_t(v) 
$$
From 
$$
\int_u^{+\infty}  dv f_t(v)= G_t(u)= G(u)^{\alpha_t}
$$
one has
$$
P(\tau=s) = \int_{-\infty}^{+\infty}  du \ f_s (u)\   G(u) ^{\displaystyle\sum_{t:\, t\ne s} \alpha_t  } 
$$
Since
$$
f_s(u) = - \alpha_s G(u)^{\alpha_s -1} G'(u)
$$
the above integral is
$$
-  \int_{-\infty}^{+\infty}  du \   \alpha_s G(u) ^{\displaystyle\sum_{t\in I} \alpha_t  -1} G'(u)= \int_0^1 \ dg\  \alpha_s \ g \ ^{\displaystyle\sum_{t\in I} \alpha_t  -1} 
$$
that gives the result.
\end{proof}
In the case of   exponential variables $Y_t= Exp(\alpha_t)$,  with $G(y) = e^{-y}$,  one has 
$$ G_t(y) = e^{-\alpha_t y} = G(y)^{\alpha_t}$$
and so  eq \ref{alphaoversum} holds, as it is well known.
In case of shifted exponentials assume that the sequence $T_t$ is increasing.  
If  $t<u$  one has for any $y\in \mathbb{R}$
$$
e^{- \lambda\ (y-T_t) \ \chi_{[T_t, +\infty)} (y)}  \le  e^{- \lambda\ (y-T_u) \ \chi_{[T_u, +\infty)} (y)}
$$
and so   the  series  $Y_t$ is stochastically increasing, but the condition  eq. \ref{galfa} of Proposition 3.1 is not satisfied.

\section{Simulated series}

An application of Proposition 2.1  concerns the case in which  no information a priori is available on  the trend,  that  has to be   estimated from the data. In this case a trend-shifted exponential model has to be first fitted and  subsequently an estimator of the minimum location will be obtained.
In view of the application to a real case, we  first check the performance of this  estimator in a simulation.
We consider two  data generating models: the first one has   a  piecewise linear  asymmetric trend as the one  considered in  Sec. 3 (eq. \ref{trend2}).

\begin{equation}
T_t= \begin{cases}
 -a(t - t_0) & t <  t_0\\
\ \  b(t -t_0) & t \ge t_0
\end{cases}
\label{lintrend}
\end{equation}
where 
$t_0=500, \ \ a=  1/300, \ \ b=1/100$.
The index range $I = 1,..., 1000$  has the typical length of the   data of the next section; the noise term is  an i.i.d. sequence $\epsilon_t = {\rm Exp(1)}$.
We have showed  in Sec 3 that in this case the estimator of the minimum location is biased.  Eq.  \ref{bias} for $a=1/300, b=1/100, \lambda=1$ gives $ \mathbb{E} (\tau)\simeq -25$.

The second example   with asymmetric exponential trends
\begin{equation}
T_t= \begin{cases}
2(e^{- a (t-t_0)} -1) & t <  t_0\\
4(1-e^{- b (t-t_0)}) & t \ge t_0
\end{cases}
\label{exptrend}
\end{equation}
where 
$t_0=500, \ \ a=  1/500, \ \ b=1/100$, 
 is a natural model in  forced linear systems and  it was used in \cite{cammarota2011-1} for cardiac series. In both cases the true minimum location is given by the parameter $t_0$.
 In  simulations and subsequent analysis we use the free statistical software  R \cite{R}.
The  location parameters  of the exponentials  can be estimated by 
$$\hat{T}_t= min\{ Y_{t-h}, ..., Y_{t+h}\}$$
where  $h$  plays the role of the  bandwidth in non parametric estimation of the trend.
The residuals are  
$$ \hat\epsilon_t = Y_t - \hat{T}_t$$
and  the  estimated  rate  parameter  $\hat\lambda$ of the exponential can be  obtained   from a fit of the histogram of  the residuals. We consider the  estimator  of  the minimum location
\begin{equation}
\hat\tau= {\rm argmin} \{ \hat {T_t}, \ \ t \in I \}
\label{tauhat}
\end{equation}
According to  eq. \ref{expB}, the  distribution   of $\hat\tau$  is 
\begin{equation}
 P(\hat\tau=s) = \int_{\hat{T}_s}^{+\infty} du\  \hat\lambda\  e^ {- \hat\lambda \displaystyle\sum_{t: \hat{T}_t \le u}( u-\hat{T}_t)}
\label{expBbis}
\end{equation}

The confidence interval  of  the true minimum location $t_0$   of level $1- \alpha$ is obtained  by the quantiles  of  the $\hat \tau$  distribution of order $\alpha/2$ and $1-\alpha/2$. 
We have simulated the  series in order to investigate the role of the free parameter $h$, for  $h$ ranging from 5 to 20. For each of 200  realizations  the  95\% confidence interval of 
the minimum location is obtained.  We have focused on two parameters: the coverage rate, computed as the fraction of the intervals that contain the true value $t_0=500$, and the length of the interval, computed as the mean of the lengths over the realizations. The intervals obtained in the simulation  were  non symmetric with respect to 500, due to the non symmetric trend. The steps of the method are illustrated in Fig. 1 in the case of the 
 exponential trend.
 The figure  shows one of the  simulated  series, a zoom to the minimum and the plots of the distribution of $\hat\tau$ and of its logarithm.  From the distribution  of  $\hat\tau$ the 0.025 and 0.975 quantiles are computed and reported on the  plot of the series to illustrate the confidence interval of the minimum location.

\begin{table}
\caption{Simulated series: coverage rate of the  minimum location confidence interval  and its  amplitude as a function of the  bandwidth $h$. }
{\begin{tabular}{ccccccc}
\hline\\
\multicolumn{7}{c} {Peacewise linear trend}\\[2ex]

 bandwidth & 5 & 8 & 11 & 14 & 17 & 20\\
 coverage rate & 0.86 &  0.88 & 0.94 &  0.92 & 0.95 & 0.98\\
 interval length &  34 &  35 &  35 & 38  &  42 & 44 \\
 \hline \\ [0.5ex]
 \multicolumn{7}{c} {Exponential trend}\\[0.5ex]\\

 bandwidth & 5 & 8 & 11 & 14 & 17 & 20\\
 coverage rate & 0.74 &  0.78 & 0.87 &  0.90 & 0.92 & 0.94\\
 interval length &  29 &  29 &  32 & 35  &  37 & 42 \\
\hline
 \end{tabular}}
\end{table}
The results  on the  role of the parameter $h$ are  summarized in Table 1. The amplitude of the confidence interval is increasing in $h$.
The choice of $h$ is subjected  to the bias-variance tradeoff, similar to the case of bandwidth in non parametric estimation of trend. If $h$ is small, say $h=5$, the trend-shift has  large variance and small bias; if $h$ is large, say  $h=20$, the opposite occurs. We also remark that  since  the parameter $\lambda$ was estimated from the data,  its  estimate is biased   due to  non stationarity.
\begin{figure}[!th]
\begin{center}
 \includegraphics[width=13cm]{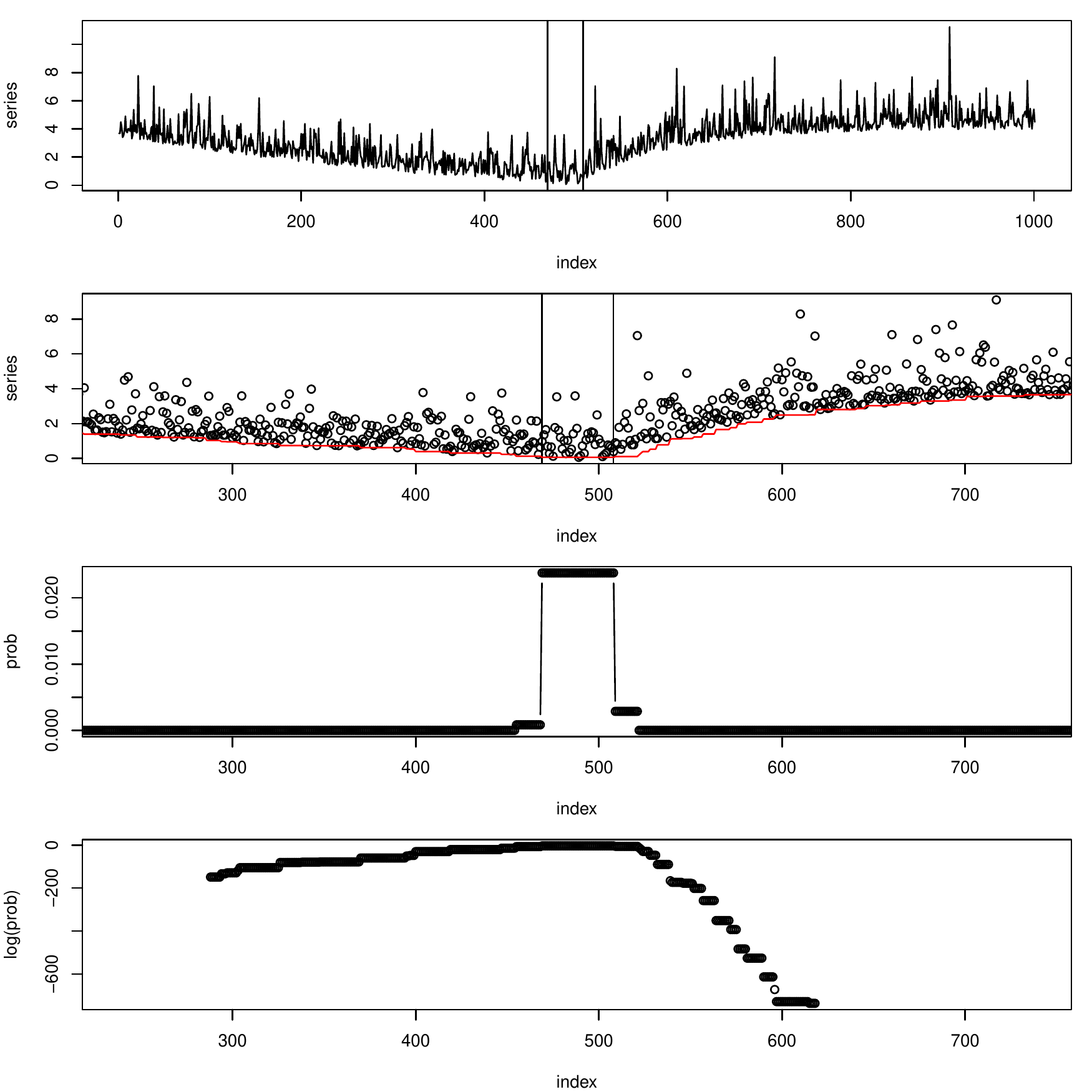}
 \end{center}
\caption{First panel: simulated   series  with  exponential trend and exponential noise; second: a zoom to the minimum and  the estimated trend (color online); third: the distribution of the minimum location; forth: the natural logarithm of the distribution. Vertical bars denote  the  95\% confidence interval of the minimum location. Bandwidth $h=20$.}
\label{fig1}
\end{figure}

\section{Real data series}
 The RR and RT intervals extracted  from the ECG  provide  respectively information on the duration of the cardiac cycle (equivalently on the heart rate) and on the repolarization phase. In the ambulatorial stress test these series show  a decreasing trend during the exercise and an increasing one  during the recovery, separated by a global minimum.  We refer to \cite{gibbons}    for  the clinical use of the test. In    \cite{cammarota2012}  the estimation of the minimum location of RR and RT series  was performed using standard methods aimed to  put into evidence  a difference in the  locations   of the two minima. These methods did not provide a confidence interval for the minimum location.   The minimum of the RR series, that corresponds to the maximum heart rate, is sharp and is affected by  low noise, while the minimum of the RT is not sharp and affected by a larger level of noise.  Hence an estimate of the confidence interval for the RT minimum 
 location is necessary to  assess the relative location of the two minima. We apply the  methods  of the previous section to an example of RT series (Fig. 2). The dependence  of the confidence interval as a function of the bandwidth parameter $h$ is shown in Table 2. 
 The intervals result to be nested increasing as $h$ increases, as it is   observed in  the simulated series.  Using the  value $h=11$  the amplitude of the  95\% confidence interval is 21 beats, providing a good localization of the minimum.

\begin{table}
\caption{The RT electrocardiographic series:  left and right  endpoints of the  95\%  confidence interval in correspondence of the bandwidth.}

\begin{tabular}{rrrrrrr}
\hline\\
bandwidth & 5 & 8 & 11 & 14 & 17 & 20 \\ 
  left end & 1148 & 1145 & 1143 & 1140 & 1137 & 1135 \\ 
  right end & 1159 & 1162 & 1164 & 1167 & 1170 & 1172 \\ 
  \hline
\end{tabular}

\end{table}

\begin{figure}[!th]
\begin{center}
 \includegraphics[width=13cm]{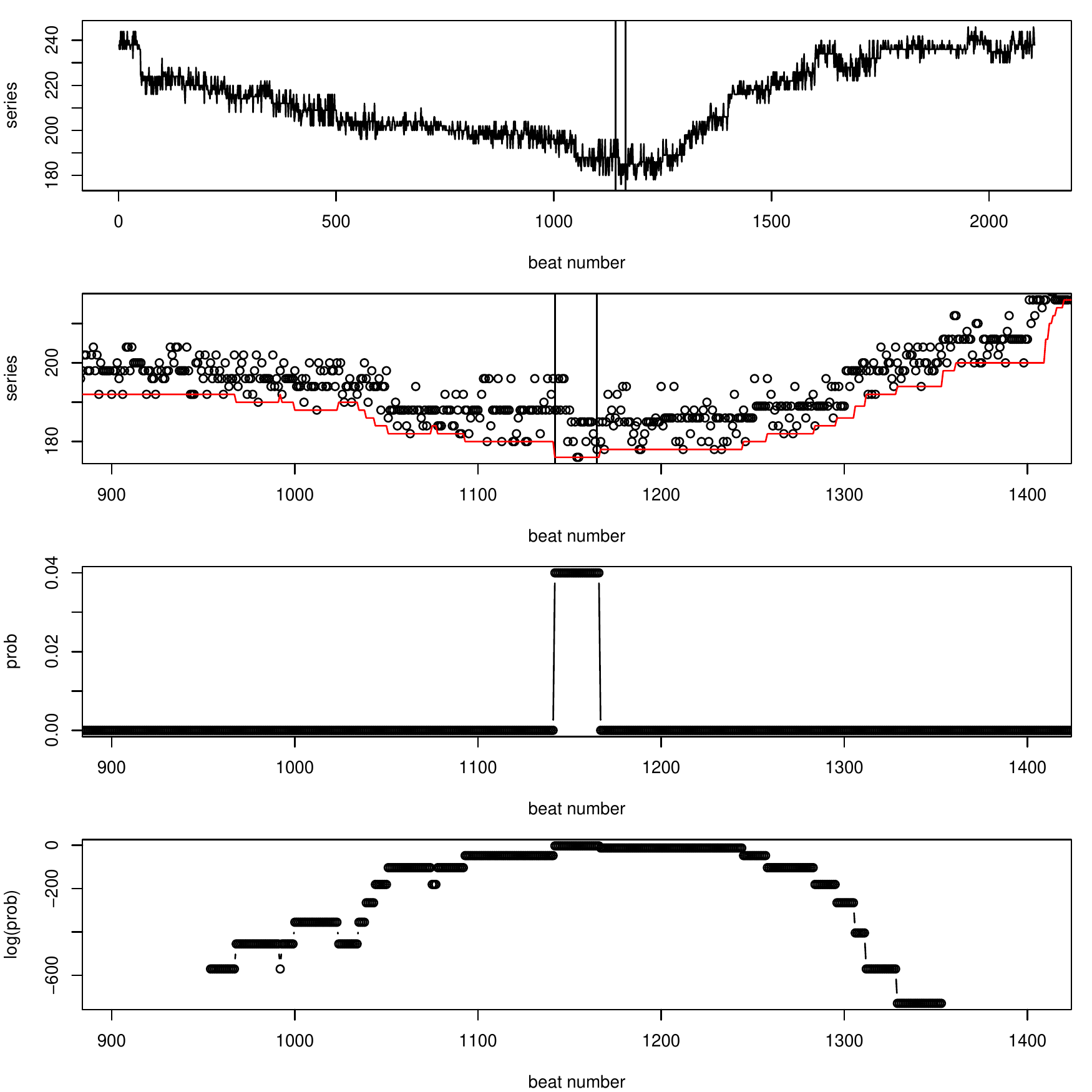}
 \end{center}
\caption{Top panel: RT electrocardiographic series   and the  95\% confidence interval of the minimum location (the  index  of x-axis is the beat number and the y-scale is millisecond); second: a zoom to the minimum and  the estimated trend (color online); third: the distribution of the minimum location and forth: its logarithm. }
\label{fig2}
\end{figure}

\section{Concluding remarks}

We have considered  time series decomposed into  deterministic trend and  serially independent  noise  with exponential distribution.  The  consideration of these series is motivated by their application to sequence of time intervals measured in cardiac recordings. 

The main result  concerns the distribution  of the minimum location. If the trend is estimated parametrically,   this distribution can be obtained exactly, via the evaluation of  an integral, as a function of the  parameters  of the trend such as slope and convexity. This approach   generalizes, at least for the  exponential noise, the one followed in \cite{mudambi, plassmann, sambucini}. We show in one example  that  if the minimum is asymmetric the expectation of the minimum location is biased. Various types of  trend having discontinuity in   first and second derivatives can be considered in future work. An open problem is the extension to Gaussian noise. 
 If the trend is  estimated non parametrically from the data our result  provides a straightforward method  to compute numerically the confidence interval of the location. This approach is new, at best of our knowledge. The investigation of the properties of the minimum location estimator  in this case is an open problem. 
 In order to check the method, we have  simulated   models  with asymmetric minimum   with  peacewise linear and exponential profiles. The coverage rates  obtained are respectively of 98\% and 
94\% for the bandwidth $h=20$. The  simulations have  put into evidence that  the   estimated confidence intervals are not symmetric with respect to the true minimum location. Since the choice of the bandwidth  is crucial, 
the role of this parameter  should be further investigated in order to  provide a criterion for  bandwidth selection.

The application concerns the   estimate of the   confidence interval of the minimum location  in  RT cardiac series  during the  stress test. In the series here analyzed,  the   estimated   95\% confidence interval has
a width of 21 beats, corresponding to about 10 seconds,  which provides a  sufficient precision.  
The application of  the  model   to  cardiac  data   series has several limitations. 
First of all our analysis revealed itself  globally effective in one case of RT series; it should be validated by 
 extensive application  to other cases, due to the large  inter individual variability.
The exponential component   is assumed to model the response of the cardiac regulation to the exercise, but  it  does not  include  the  measurement noise.  The last one  has a minor impact on the global value of the series, but this point should be further investigated. The  assumption that  the stochastic component  in RT series  is  an exponential requires validation, since previous informations  were concerning the RR series.
We have noticed that the exponential fit of the residual is poor  if  the bandwidth $h$ is large, as expected.
The  time intervals are measured  by the   ECG recorder  with a resolution of 2 milliseconds so that they  are  only approximately described  by  a continuous variable.  
 However the focus of the present study is on the trend component  of the series, which  during the stress test   plays a dominant role with  respect to the  random one. 
Finally the serial independence  should be also  carefully investigated. Testing serial independence of residuals   in cardiac series after detrending    requires non parametric methods as the one considered in \cite{cammarota2011-2}.

\section{Funding}
MIUR, Italian Ministry of Instruction, University and Research.

\bibliography{cammarota.bib}

\end{document}